%% file: balanced-separator.tex
\documentclass[11pt,twoside]{article}

\usepackage{jeffe,graphicx,hyperref}
\usepackage{url}
\usepackage{graphicx}
\usepackage[utf8]{inputenc}
\usepackage[margin=1in]{geometry}
\usepackage{marvosym}
\usepackage{cite}
\usepackage{verbatim}
\usepackage{color}
\usepackage{comment}
\usepackage{algorithm}
\usepackage{subfigure}
\usepackage{enumitem}
\usepackage[noend]{algorithmic}

\def\compilelong{}
\ifx\compilelong\undefined
  \excludecomment{longversion}
  \newenvironment{shortversion}{}{}
\else
  \excludecomment{shortversion}
  \newenvironment{longversion}{}{}
\fi

\def\etal{\emph{et~al.}}			
\def\arcto{\mathord\shortrightarrow}

\def\head{\mathit{head}}
\def\tail{\mathit{tail}}

\def\rev{\mathit{rev}}

\def\Real{\mathbb{R}}

\let\eps\varepsilon
\def\cost{\mathit{c}}
\def\weight{\mathit{w}}
\def\ratio{\mathit{r}}
\def\dist{\mathrm{dist}}
\def\opt{\mathrm{OPT}}
\def\concat{\circ}

\newtheorem{theorem}{Theorem}[section]

\newtheorem{lemma}[theorem]{Lemma}

\begin{document}

\pagestyle{myheadings}
\markboth{A Polynomial-time Bicriteria Approximation Scheme for Planar Bisection}
		 {Kyle Fox, Philip N. Klein, and Shay Mozes}

\begin{titlepage}

\title{A Polynomial-time Bicriteria Approximation Scheme for Planar Bisection}
\author{
Kyle Fox\thanks{Duke University, Durham, NC, \url{kylefox@cs.duke.edu}.
  Portions of this work were done while the author
  was a postdoctoral fellow at the Institute for Computational and
  Experimental Research in Mathematics, Brown University, Providence,
  RI. Partially funded by NSF grants CCF-09-40671 and CCF-10-12254.}
  \and
Philip N. Klein\thanks{Brown University, Providence, RI, \url{klein@brown.edu}.
Research funded by NSF grants CCF-09-64037 and CCF-14-09520.}
\and
Shay Mozes\thanks{Interdisciplinary Center, Herzliya, Israel, \url{smozes@idc.ac.il}.
  Partially supported by an Israeli Science Foundation grant ISF
  749/13, and by the Israeli ministry of absorption.}
}

\maketitle
\begin{abstract}
Given an undirected graph with edge costs and node weights, the minimum bisection problem asks for a partition of the nodes into two parts of equal weight such that the sum of edge costs between the parts is minimized. We give a polynomial time bicriteria approximation scheme for bisection on planar graphs.

Specifically, let~$W$ be the total weight of all nodes in a planar graph~$G$. For any constant $\eps > 0$, our algorithm outputs a bipartition of the nodes such that each part weighs at most $W/2 + \eps$ and the total cost of edges crossing the partition is at most $(1+\eps)$ times the total cost of the optimal bisection. The previously best known approximation for planar minimum bisection, even with unit node weights, was~$O(\log n)$.
Our algorithm actually solves a more general problem where the input may include a target weight for the smaller side of the bipartition.
\end{abstract}


\thispagestyle{empty}
\setcounter{page}{0}
\end{titlepage}

\input{intro}

\input{prelims}
\input{skeleton}

\input{simple-shortcuts}

\input{double-cover}
\input{clustering}
\input{spanner}
\input{ack}

\bibliographystyle{abbrv}
\bibliography{short,balanced-separator}

\end{document}

%% file: intro.tex
\section{Introduction}

Breaking up is hard to do.  The most famous hard graph-breaking
problem is {\em graph bisection}: partitioning the vertices of a graph into two
equal-size subsets so as to minimize the number of edges between the
subsets.  This problem was proved NP-hard by Garey, Johnson, and
Stockmeyer~\cite{GareyJS76} in 1976.  

\paragraph{Background}
But how hard is it really?  In particular, how well can graph
bisection be approximated in polynomial time?  Even assuming $P\neq
NP$, we cannot at this point rule out the existence of a
polynomial-time approximation scheme (PTAS) for graph bisection.  The
best approximation ratio known to be achievable in polynomial time is
$O(\log n)$ where $n$ is the number of vertices, due to
R\"acke~\cite{Racke08}, improving on a bound of $O(\log^{1.5} n)$ due
to Feige and Krauthgamer~\cite{FeigeK02} (the first result discovered
that had a
polylogarithmic approximation ratio).

One way to make the problem easier is to relax the balance condition.
Given a number
$0<b\leq 1/2$, a bipartition $U\cup V$ of a graph's vertices is {\em
  $b$-balanced}\footnote{Some papers use this term to mean that each
  part has cardinality {\em at most} $bn$.} if $|U|\geq \lfloor b n\rfloor$ and $|V|\geq \lfloor b
n\rfloor$.    Any bipartition $U\cup V$ induces a {\em cut}, namely the set
of edges between $U$ and $V$.  The bisection problem is to find a minimum
$\frac{1}{2}$-balanced cut.  It might be a simpler problem to find a
nearly optimal
$b$-balanced cut for some $b<1/2$.  No better
approximation ratio is known for this problem when the input graph is
arbitrary.  However, for the special case of {\em planar} graphs, a
2-approximation algorithm was given by Garg, Saran, and
Vazirani~\cite{GargSV99} for finding a minimum $b$-balanced cut for
any $b\leq 1/3$.

Bicriteria approximation\footnote{This approach is also called a
  pseudo-approximation.} gives another way to relax the balance
condition.  A bicriteria approximation algorithm seeks a $b'$-balanced
cut whose size is at most some factor times the minimum size of a
$b$-balanced cut.  In an early and very influential paper, Leighton
and Rao~\cite{LeightonR99} showed, using a reduction to their $O(\log
n)$-approximation algorithm for another problem, {\em uniform sparsest
  cut}, that a $b'$-balanced cut could be found whose size is
$O(\frac{\log n}{b-b'})$ times the minimum size of a $b$-balanced cut
for any $b'<b$ and $b'<1/3$.  Using the improved $O(\sqrt{\log
  n})$-approximation algorithm of Arora, Rao, and
Vazirani~\cite{AroraRV09} for uniform sparsest cut, the approximation
ratio $O(\frac{\log n}{b-b'})$ can be improved to $O(\frac{\sqrt{\log
    n}}{b-b'})$.
Note that this performance ratio gets worse as the graph
size grows and gets worse as the balance $b'$ achieved by the
algorithm approaches the balance $b$ that defines the optimum.

\paragraph{Our results}
In this paper, we give a {\em bicriteria approximation scheme} for
bisection in planar graphs:

\begin{theorem} \label{thm:main} For any $\epsilon>0$, there is a
  polynomial-time algorithm that, given a planar graph $G$, returns a
  $\frac{1}{2}-\epsilon$-balanced cut whose size is at most
  $1+\epsilon$ times the optimum bisection size.
\end{theorem}

That is, the algorithm returns a partition of the vertex set that is
{\em almost} perfectly balanced (each side has at most a fraction
$\frac{1}{2}+\epsilon$ of the vertices) and whose size is almost as small as the
smallest bisection.

Previously {\em no approximation algorithm was known that had a
$1+\epsilon$ approximation ratio} even if the algorithm was allowed to
return a $b'$-balanced cut for some constant $b'>0$, even if $b'$ was
allowed to depend on $\epsilon$, even for planar graphs.
The algorithm generalizes to handle $b$-balanced cuts:

\begin{theorem} \label{thm:main2} For any $\epsilon>0$, there is a
  polynomial time algorithm that, given a planar graph $G$ and given $0<b\leq\frac{1}{2}$, returns
  a $(b-\epsilon)$-balanced cut whose size is at most $1+\epsilon$ times the minimum $b$-balanced
  cut.
\end{theorem}

Also, the algorithm can handle nonnegative costs on edges and
nonnegative weights on vertices.  In fact, we prove a more powerful
theorem.  Let $G$ be a graph with vertex-weights and edge-costs.  For
a number $b$, a {\em $b$-bipartition} of $G$ is a bipartition $U \cup
V$ of the vertices of $G$ such that the total weight of $U$ is exactly
$b$ times the total weight of $G$.  The {\em cost} of the bipartition
$U\cup V$ is the cost of the corresponding cut, i.e. the set
of edges connecting $U$ and $V$.  For example, a minimum bisection is
a minimum-cost $\frac{1}{2}$-bipartition where all costs and weights
are~1. Let $\opt_b(G)$ be the cost of an optimal $b$-bipartition.

\begin{theorem} \label{thm:main3}  For any $\epsilon>0$, there is a
  polynomial-time algorithm that, given $b\geq 0$ and given a planar
  graph $G$ with edge costs and vertex weights such that $G$ has a
  $b$-bipartition, returns a $b'$-bipartition whose cost is at
  most $(1+\epsilon) \opt_b(G)$, where $b'\in [b-\epsilon, b+\epsilon]$.
\end{theorem}

\subsection{Related work}

There is much prior work on finding approximately optimal separators
in planar graphs.  Before the work of Leighton and
Rao~\cite{LeightonR99}, Rao~\cite{Rao87,Rao92} gave
approximation algorithms for balanced cuts in {\em planar} graphs.  
One is a true approximation algorithm for finding a $b$-balanced cut
(for $b\leq \frac{1}{3})$ whose performance guarantee is
logarithmic.  Another is a bicriteria approximation algorithm whose performance guarantee
grows as the balance $b'$ achieved by the algorithm approaches the
balance $b$ that defines the optimum.  In our algorithm, we employ a
subroutine of Rao~\cite{Rao92}.

Park and Phillips~\cite{ParkPhillips} gave improved algorithms for
achieving some of the goals of~\cite{Rao87,Rao92}.  The aforementioned
result of Saran, Garg, and Vazirani~\cite{GargSV99}, the
2-approximation algorithm for $b$-balanced cut in planar graphs for
$b<1/3$, built on the work of Park and Phillips.

There has been work on approximation schemes for other graph classes.
Arora, Karger, and Karpinski~\cite{AroraKK99} gave an approximation
scheme for bisection in unit-edge-cost {\em dense} graphs, graphs with
$\Omega(n^2)$ edges.  Guruswami, Makarychev, Raghavendra, Steurer and
Zhou~\cite{GuruswamiMRSZ11} gave an algorithm that, given a graph in which there is a
bisection that cuts a fraction $1-\epsilon$ of the edges, finds a
bisection that cuts a fraction $1-g(\epsilon)$ of the edges, where
$g(\epsilon)=O(\sqrt[3]{\epsilon} \log(1/\epsilon))$.

There has been much work on using approximation algorithms for
balanced separators to obtain approximation algorithms for other
problems; see the survey of Shmoys~\cite{ShmoysCutProblems}.

There has been much work on finding balanced separators of size
$O(\sqrt{n})$ in planar graphs, regardless of the optimum bisection
size.  Lipton and Tarjan~\cite{LT79} first gave such an algorithm.
Many papers built on this result: improvements to the multiplicative
constant, an algorithm to find separators that are also simple
cycles~\cite{Miller86}, algorithms that build on geometric embeddings~\cite{MillerTTV97}
or eigenvectors~\cite{SpielmanT96}.  In addition, 
there has been much work on algorithms that use planar separators.


\section{Overview}

We outline the algorithm for Theorem~\ref{thm:main3}, the bicriteria
approximation scheme for min $b$-bipartition.  The input is a
graph $G^*$ with vertex weights and edge costs.  The algorithm involves
edge {\em contraction}.  Contracting an edge $uv$ means removing the edge and replacing its
endpoints $u$ and $v$ with a single vertex $x$.  Edges previously
incident to $u$ or $v$ are now incident to $x$.  The weight assigned
to $x$ is the sum of the weights of $u$ and $v$.

\subsection{Framework for approximation scheme}
\label{sec:intro:framework}

The algorithm uses a framework of
Klein~\cite{Klein08}.  The framework has previously been used to
address optimization problems in planar
graphs~\cite{BateniCEHKM11,BateniHM11,MultiwayCut2012,BG07,BorradaileKleinMathieu2009,Klein08},
such as traveling salesman and Steiner tree, that involve minimizing the
cost of a set of edges subject to {\em connectivity} constraints.
The framework has never been used before in the context of {\em
  weight} constraints on the vertices.

First we give the outline.  Fix $0\leq b \leq 1$ and $\epsilon>0$.
The framework uses the notion of {\em branchwidth}\footnote{
{\em Treewidth} could be used instead}~\cite{ST94,Klein08}.

\begin{enumerate}
\itemsep=0pt
\item {\em Spanner step:} In the input graph $G^*$, contract a selected set of edges, 
  obtaining a graph $\widehat{G^*}$ with the following properties:
\begin{itemize}\itemsep=0pt
\item $\cost(\widehat{G^*}) \leq \rho \opt_b(G^*)$, and
\item there exists $b'\in [b-\epsilon,b+\epsilon]$ such that $\opt_{b'}(\widehat{G^*}) \leq (1+\epsilon) \opt_b(G^*)$.
\end{itemize}
where $\rho$ is a quantity that depends on the construction.\\
\item {\em Thinning step:} Select a set $S$ of edges in $\widehat{G^*}$
  such that:
\begin{itemize}\itemsep=0pt
\item 
$\cost(S) \leq (1/k) \cost(\widehat{G^*})$ and
\item 
$\widehat{G^*}-S$ has branchwidth $O(k)$,
\end{itemize}
where $k = \epsilon^{-1} \rho$.
\item {\em Dynamic-programming step:} Find a cheapest 
  $b'$-bipartition $(\widehat{U_1}, \widehat{U_2})$ in
  $\widehat{G^*}-S$, where $b' \in [b-\epsilon,b+\epsilon]$.
\item {\em Lifting step:}  Return $(U_1, U_2)$ where $U_i$ is the set of vertices of $G^*$
  coalesced to form vertices in $\widehat{U_i}$.
\end{enumerate}
The cost of the returned solution is at most
\begin{eqnarray*}
\opt_{b'}(\widehat{G^*}-S)+\cost(S) & \leq & \opt_{b'}(\widehat{G^*})+\cost(S)\\
& \leq & (1+\epsilon)\opt_b(G^*) + (1/k)\cost(\widehat{G^*})\\
& \leq & (1+\epsilon)\opt_b(G^*) + \epsilon \rho^{-1} \cost(\widehat{G^*})\\
& \leq & (1+\epsilon)\opt_b(G^*) + \epsilon \opt_b(G^*)
\end{eqnarray*}
which shows that the cost is at most $1+2\epsilon$ times the cost of
an optimal $b$-bipartition.  The fact that $(\widehat{U_1},
\widehat{U_2})$ is a $b'$-bipartition of $\widehat{G^*}$ means
that $(U_1, U_2)$ is a $b'$-bipartition of $G^*$.

The thinning step is straightforward: choose a root, and find breadth-first-search levels
for all edges in $\widehat{G^*}$.  For $i=0,1,2,\ldots,
k-1$, let $S_i$ be the set of edges whose levels are congruent mod~$k$
to $i$.  For each $i$, $\widehat{G^*}-S_i$ has branchwidth $O(k)$ (see,
e.g.,~\cite{Klein08}, also treewidth $O(k)$, see, e.g.,~\cite{Baker94})
 and at least one of the sets $S_i$ has cost at most $(1/k) \cost(\widehat{G^*})$.

The fact that $\widehat{G^*}-S$ has branchwidth $O(k)$ means that in the
dynamic-programming step an
optimal $b'$-bipartition can be found in time $2^{O(k)}\text{poly}(n, W)$ where $W$ is the sum of weights.

Assume for now that $W$ is $O(n)$ and that $\rho$ is $O(\log n)$.
Then the dynamic-programming step requires only polynomial time. 
(This is a straightforward generalization of Theorem~4.2
of~\cite{JansenKLS05}.)

The one challenging step is the {\em spanner step}.\footnote{This is
 usually the case in applications of the framework.}  The main work of
this paper is showing that this step can be done.

\begin{theorem} \label{thm:spanner} There is a constant $c$ and a polynomial-time
  algorithm that, given $\epsilon>0$, $b>0$ and a planar embedded graph $G^*$ with
  vertex weights and
  edge-costs, returns a graph $\widehat{G^*}$, obtained from $G^*$ by
  contracting edges, with the following properties:
  $\cost(\widehat{G^*}) \leq c \log n \cdot \opt_b(G^*)$, and
$\exists b'\in [b-\epsilon,b+\epsilon]$ such that
 $\opt_{b'}(\widehat{G^*}) \leq (1+\epsilon) \opt_b(G^*)$.
\end{theorem}

Once we have proved Theorem~\ref{thm:spanner}, showing
that there is a poly-time algorithm for
the {\em spanner} step of the framework, we will have proved
Theorem~\ref{thm:main3}.

\subsection{Spanner construction overview}


Many tools have been developed for spanner construction.  One
tool first used for Steiner TSP is this {\em boundary-to-boundary
  spanner}:
\begin{lemma}[Theorem 6.1 of \cite{Klein06}] \label{lem:boundary-spanner} Let $G$ be a planar
  embedded graph with edge-costs, and let $C$ be the boundary of some face of $G$.
  For any $\epsilon>0$, there is a subgraph $H$ of cost
  $O(\epsilon^{-4} \cost(C))$ such that, for any two vertices $x$ and
  $y$ of $C$, the $x$-to-$y$ distance in $H$ is at most $1+\epsilon$
  times the $x$-to-$y$ distance in $G$.  Furthermore, there is an $O(n
  \log n)$ time algorithm to derive $H$ from $G$.
\end{lemma}

The edges defining the min-cost bisection or $b$-bipartition of input graph $G^*$ correspond in
the planar dual $G$ to a collection of edge-disjoint cycles.
Fragments of these cycles are paths; perturbing the solution by replacing such a path with a
nearly shortest path in a subgraph $H$ does not increase the cost of the solution by much.  This
simple idea is at the heart of the spanner construction.

At the highest level, we use the following strategy:
(Step 1) select a collection of
cycles, (Step 2) join some of them together with paths, (Step 3) for each region
of the planar dual bounded by these cycles and paths, for each
cycle $C$ that forms part of the boundary of that region, construct
the boundary-to-boundary spanner for $C$-to-$C$ paths in that
region.
The union of edges from Steps 1, 2, and 3 form the spanner.  

So far, however,  we have not handled weights.  Indeed, a
perturbation (in which a path $P$ of the optimal solution is replaced with
a path $P'$ in the spanner) could shift weight from one side of the
bipartition to the other.
We need a way to limit the amount of weight that could shift.
This is the purpose of Step~1.  Note that the original path $P$ and
replacement path $P'$ form a cycle $C$, and that weight that could
shift is enclosed by $C$.  The goal of Step~1 is to
ensure that, for every such cycle $C$ derived from such a
perturbation, 
the weight enclosed by $C$ is small
compared to the cost of $C$.  That way, the total amount of weight
shifted can be charged to the cost of the optimal solution.

Step~1 ensures that such cycles' cost/weight ratios are large by
greedily finding a collection of mutually noncrossing cycles of small
ratio.  Once Step~1 has completed, each of the regions bounded by
the noncrossing cycles contains no cycle with small cost/weight ratio (essentially).

The fact that cycles found in Step~1 have small
cost/weight ratio is used to show that the total cost of all those
cycles is not much more than the cost of the optimal solution.
Each cycle's cost is charged to the weight of {\em some} of the faces it encloses.
If we ensure that each face is charged at most a logarithmic number of times,
the total cost of the cycles is at most a log times $\opt$.

To ensure logarithmic charging, Step~1 alternates between adding cycles
to the spanner and removing cycles.  When one cycle is enclosed by the
other but the two cycles enclose almost the same weight, the pair of
cycles is designated a {\em splicing pair}, and, in an operation
called {\em splicing}, cycles
sandwiched between them are removed from the spanner. 

Here is one complication: A region bounded by cycles from Step~1 often is
bounded by {\em several} cycles, i.e. its boundary is disconnected.
The boundary-to-boundary spanner (Lemma~\ref{lem:boundary-spanner})
works only for a connected boundary.  Step~2 therefore uses a technique
called {\em PC clustering}, due to Bateni,
Hajiaghayi, and Marx~\cite{BateniHM11}, to add paths joining some of
the boundary components.  If two boundary components remain
unconnected after PC clustering, we can assume that some near-optimal solution does not connect them.

Here is another complication:  Consider a cycle $C$ associated with a
perturbation that replaces a path
$P$ of the optimal solution with a path $P'$ in the spanner. If $C$
happens to be sandwiched between two cycles comprising a splicing
pairs, then $C$ might have small cost/weight ratio. This happens
if $C$ encloses the inner cycle $C'$ of the splicing pair. To make sure
no such cycle $C$ is used in a perturbation, Step~3, in forming the boundary-to-boundary
spanners for the region $R$ containing $C'$, distinguishes between
paths going clockwise around $C'$ and paths going counterclockwise.  This is accomplished using a construction
from topology, the {\em cyclic double cover}.

%% file: prelims.tex
\section{Preliminaries}
\label{sec:prelim}

\paragraph{Achieving polynomially bounded weights}\
Garg, Saran and Vazirani~\cite{GargSV99} observed that if one is willing
to accept a $(b \pm \eps)$-bipartition instead of a $b$-bipartition, then
one can assume polynomially bounded weights. This is done by defining
new weights $\weight(v) \leftarrow \lfloor \weight(v) \frac{n}{\eps
  W} \rfloor$ where $W$ is the sum of original weights and $n$ is the
number of vertices of the input graph.  After the transformation, the sum of weights is bounded
by $\eps^{-1}n$. However, due to the truncation, the weight of any
vertex may be off by $\frac{\eps W}{n}$ with respect to its original
weight.
Therefore, the weight of any set in a bipartition may be off by at most $\eps W$.
This is allowed by our theorems.  We therefore assume henceforth that
the sum of weights is polynomially bounded.

\paragraph{Basic definitions}

We assume that the reader is familiar with the basic concepts of planar
graphs such as planar embeddings and planar duality.
We use $G^*$ to denote the planar embedded input graph, and we use $G$ to denote its
planar dual.  The costs of
edges in $G^*$ are assigned to the corresponding edges of $G$. 
The weight function
can be viewed as an assignment of weights to faces.
For
the remainder of this paper we deal with the dual graph
$G$.  Henceforth, unless
otherwise stated, vertices, faces and edges refer to those of $G$.

For each edge $e$ in the edge-set $E$, 
we define two oppositely directed darts, one in each orientation. 
We define $\rev(\cdot)$ to be the function that maps each dart to the corresponding dart in 
the opposite direction.

A non-empty sequence 
$d_1 \ldots d_k$
 of darts is a {\em walk} if
the head of
$d_i$ is the tail of $d_{i+1}$ for every $1 \leq i \leq k$.
A walk is said to
be a {\em closed walk} if the tail of $d_1$ is the head of $d_k$.


Let $X$ be a walk in a planar embedded graph, and let $P=a\ X\ b$ and $Q=c\ X\ d$ be walks
that are identical except for their first and last darts.  
Let $a'$ be the successor of $a$ in $P$ and let $b'$ be the
predecessor of $b$ in $P$.  We say $Q$ forms a {\em crossing
  configuration} with $P$  if the clockwise cyclic order of darts whose head is
$\head(a)$ induces the cycle $(a\ \rev(a')\ c)$ and the clockwise
cyclic order of darts whose tail is 
$\tail(b)$ induces the cycle $(b\ \rev(b')\ d)$. 

We say a walk $P$ {\em crosses} a walk $Q$ if a subwalk of $P$ and a subwalk
of $Q$ form a crossing configuration.  
See Figure~\ref{fig:crossing-paths}.  We define a {\em cycle (of darts)} to be a
non-self-crossing closed walk that uses each dart at most once.  We
omit the modifier ``of darts.''
Thus for our purposes a cycle can use an edge at most twice--once in each direction---and
cannot cross itself.


We can assume (by adding a self-loop if necessary) that $G$ has a face
with zero weight whose boundary has zero cost.  We use $f_\infty$ to
denote this face, and we refer to it as the infinite face.

Let $C$ be the dart multiset of a set of closed walks, and let $C^*$ be the
multiset of corresponding darts in $G^*$.
Let $f$ be a face of $G$.
Let $P$ be any $f_\infty$-to-$f$ path in $G^*$ (the input graph).
We say $f$ is {\em enclosed} by $C$ if the parity of 
$|P\cap C^*|-|P\cap \rev(C^*)|$ is odd. See Figure~\ref{fig:enclosure}.

\begin{figure*}[t]
\begin{center}
\subfigure[Two crossing paths and a cycle of darts (left), a set of mutually
noncrossing cycles (middle), and a set of cycles that are not mutually noncrossing
 (right).  In the middle, the region bounded by three of the cycles is shaded.\label{fig:crossing-paths}]
{        \includegraphics[scale=1.0]{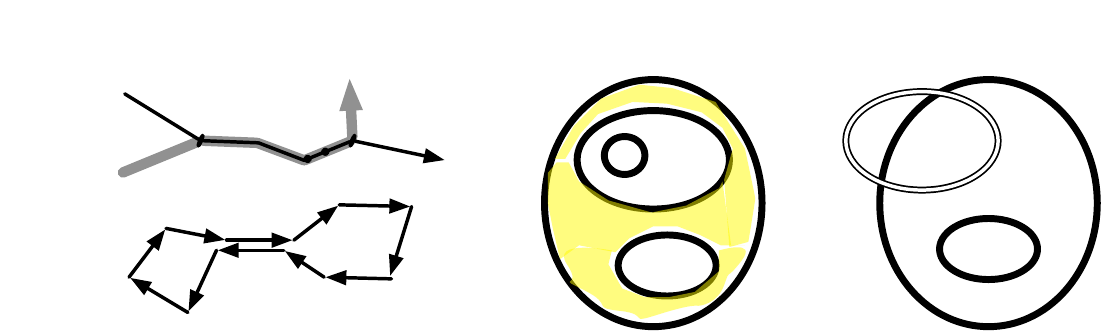}}
\hspace{0.05\textwidth}
\subfigure[A set $\mathcal C$ of two cycles (double lines). Face $a$ is
          enclosed by $\mathcal C$, but face $b$ is not. \label{fig:enclosure}]
{\includegraphics[scale=0.4]{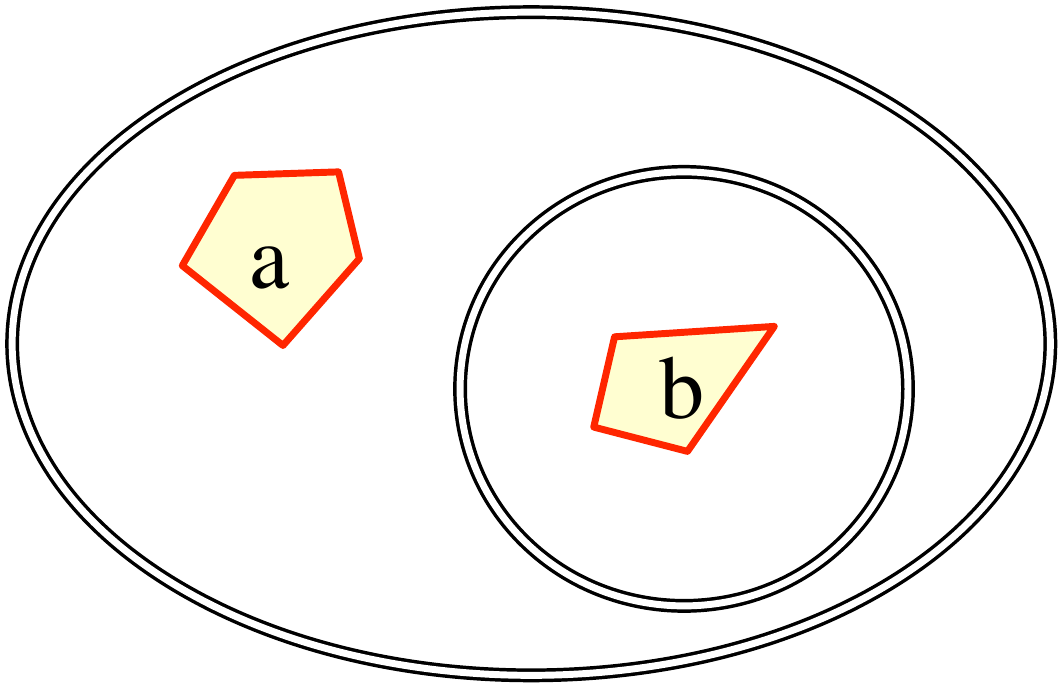}}
\caption{Crossings and enclosure}
\end{center}
\end{figure*}

Consider a bipartition $U\cup V$ in $G^*$ where the vertex of $G^*$
corresponding to $f_\infty$ belongs to $V$.  The duality of cuts and
cycles implies that the edges crossing this cut form a set
$\mathcal C$ of cycles in $G$.  The weight of the bipartition in $G^*$
is the total weight of faces enclosed by $\mathcal C$.

For a cycle $C$ of $G$,
we define  $\weight_G(C)$ to be the sum of weights of
faces enclosed by $C$ in $G$.  For a set $S$ of faces of $G$, we define
$\weight_{S}(C)$ to be the sum of weights of faces in $S$ enclosed by
$C$, and
we define $\ratio_{S}(C)$ to be the ratio of the cost of $C$ to $\weight_{S}(C)$.

\paragraph{Regions defined by a mutually non-crossing set of cycles}
Let $\mathcal C$ be a set of cycles that are mutually
non-crossing. See Figure~\ref{fig:crossing-paths}.
Assume for the sake of convenience that the boundary of $f_\infty$ is in
$\mathcal C$.
The set $\mathcal C$ can be represented by a rooted ordered tree $\mathcal
T$. Every node $v \in \mathcal T$ corresponds to a
cycle $C_v \in \mathcal C$. 
Ancestry in $\mathcal T$ is determined by enclosure: 
Node $v$ is an ancestor of node $u$ in $\mathcal T$ if $C_v$ encloses
$C_u$ in $G$. Thus, $f_\infty$ is the cycle corresponding to the
root of $\mathcal T$. 

Every node $v$ of $\mathcal T$ is associated
with a \emph{region} $R_v$. $R_v$ is the subgraph of $G$
consisting of vertices, edges, and faces enclosed by $C_v$, 
and not strictly enclosed by $C_u$ for any child $u
\in \mathcal T$ of $v$.

The cycle $C_v$ is called the \emph{outer
  boundary} of $R_v$. For a child $u$ of $v$, the cycle $C_u$ is
called a \emph{hole} of $R_v$.  The {\em weight} of a hole is the
total weight of faces enclosed by the hole.
Together, the outer boundary and the
holes form the {\em boundary} of $R_v$.  We say $R_v$ {\em strictly contains} an edge
if in addition the edge does not belong to the boundary of $R_v$.

We say that a region $R$ {\em contains} a cycle $C$ if $R$ contains every edge of
$C$, and that $R$ {\em strictly contains} $C$ if in addition $R$
strictly contains at least one edge of $C$.

\paragraph{Finding low-ratio cycles}\
Let $T$ be a shortest path tree with root $r$. Let $C$ be a 
cycle that encloses $r$.
We say that $C$ is discovered by $T$ from
the inside if, for every $v \in C$, the $r$-to-$v$ path in $T$ is
enclosed by $C$. Rao~\cite{Rao92} described a polynomial time technique that finds,
for every possible weight $w$, the minimum-cost cycle enclosing
exactly $w$ weight among cycles that go through $r$ and are discovered
by $T$ from the inside (if such a cycle exists).
\begin{shortversion}
Rao's technique can be used to find a maximally face-enclosing cycle with ratio at most some ratio
$\alpha$.
See the full version for details.
\end{shortversion}

\begin{longversion}
The following lemma implies that Rao's technique can be used 
to find a maximally
face-enclosing cycle with ratio at most some ratio $\alpha$.
\begin{lemma}\label{lem:maximal-discovered-from-inside}
Let $G$ be a planar embedded graph with edge costs and face weights.
Let $T$ be a shortest path tree in $G$, rooted at a vertex $r$. 
Let $C$ be  a maximally face-enclosing cycle $C$  with ratio at most
$\alpha$.
If $C$ encloses $r$ then $T$ discovers $C$ from the inside.
\end{lemma}
\begin{proof}
Assume that $C$ is not discovered by $T$
from the inside.
Then there is a vertex $v \in C$ such that the
$r$-to-$v$ path $P$ in $T$ is not enclosed by $C$. Let $P'$ be a
maximal subpath of $P$ consisting only of edges that are not enclosed
by $C$. Let $x,y$ be the endpoints of $P'$. Note that, since $C$
encloses $r$, both $x$ and $y$ are vertices of $C$. Let $Q'$ be a subpath of
$C$ with endpoints $x$ and $y$ such that the cycle $C' = Q' \circ P'$ encloses
$C$. Since $P'$ is a shortest $x$-to-$y$ path, the cost of $C'$ is at most the
cost of $C$.
Cycle $C'$ also encloses every face enclosed by $C$, so the weight of $C'$ is at least the weight
of $C$. This contradicts the fact that $C$ is  a maximally face-enclosing cycle $C$  with ratio at most
$\alpha$.
\end{proof}
To find a maximally face enclosing cycle with ratio at most $\alpha$, 
consider the set $M$ of cycles whose weight is maximum among all cycles
whose ratio is at most $\alpha$. Let $C$ be a cycle in $M$ enclosing the
greatest number of faces (there may be faces with zero weight). 
Note that $C$ is a maximally face enclosing cycle with ratio at most
$\alpha$.
Let $w$ denote the weight of $C$.
To find $C$, slightly perturb the weight of
every face to make it non-zero without significantly changing the
total weight of any cycle. This can be done by scaling the weights by
the number of faces, and adding 1 to the weight of every face. Note
that this transformation keeps the total weight polynomially bounded.
For every possible choice of root $r$ of the shortest path tree
$T$, use Rao's technique to compute, for each
possible (perturbed) weight $x$, the minimum cost cycle discovered from the
inside, and enclosing exactly $x$ perturbed weight. Return the cycle enclosing
maximum perturbed weight whose unperturbed ratio is at most $\alpha$.

Let $w'$ be the perturbed weight of $C$.
By Lemma~\ref{lem:maximal-discovered-from-inside}, for a correct
choice of~$r$,  $C$ is the minimum-cost cycle computed for
weight $w'$. 
By definition of the perturbation, the perturbed weight $w'$ of $C$ is greater than that of any
other cycle with ratio at most $\alpha$ that encloses fewer faces than
$C$. Also, no cycle whose (perturbed) weight is greater than $w'$ has
ratio at most $\alpha$ with respect to unperturbed weights. Therefore,
the procedure described returns the cycle $C$.

\end{longversion}

\paragraph{Edge contractions}\
One step of the algorithm described in this paper uses edge contractions. The
planar embedding is not relevant in this step. We therefore use a
definition of edge contraction that does not depend on the
embedding. Contracting an edge $e$ results in merging the endpoints of $e$
into a single vertex. Any self loops that are created in the process  are deleted.

%% file: skeleton.tex
\section{Skeleton construction}

We now begin to describe our procedure for constructing a spanner to approximate
a minimum cost $b$-bipartition.
For brevity, we let $\opt$ be the cost of an optimal $b$-bipartition.
Let~$W$ be the total weight of all faces in~$G$.
Recall that we work directly with the dual graph~$G$ of our input graph~$G^*$ and that a solution or
$b$-bipartition is a set of cycles in~$G$.
We assume that the spanner algorithm is given a rational number $\lambda$ such
that $\opt / W \leq \lambda \leq 2 \opt / W$, since an outer loop
surrounding the approximation algorithm can try different values of
$\lambda$ and return the best solution obtained.

The algorithm constructs a family~$\mathcal C$ of mutually noncrossing
cycles, which includes  the outer
face~$f_\infty$.  We refer to $\mathcal C$ as the \emph{skeleton}.
As discussed in Section~\ref{sec:prelim}, 
the cycles of $\mathcal C$ define {\em regions}, and 
define a rooted tree $\mathcal{T}$
based on enclosure.  For each node $v$ of this tree, we order the
children of $v$ left to right according to increasing weight of the
faces they enclose. See Figure~\ref{fig:preorder}. 
A preorder traversal of $\mathcal T$ that visits
siblings in this order defines a total left-to-right ordering on the
nodes of $\mathcal T$. 

\begin{figure}[t]
\centering
\includegraphics[scale=.3]{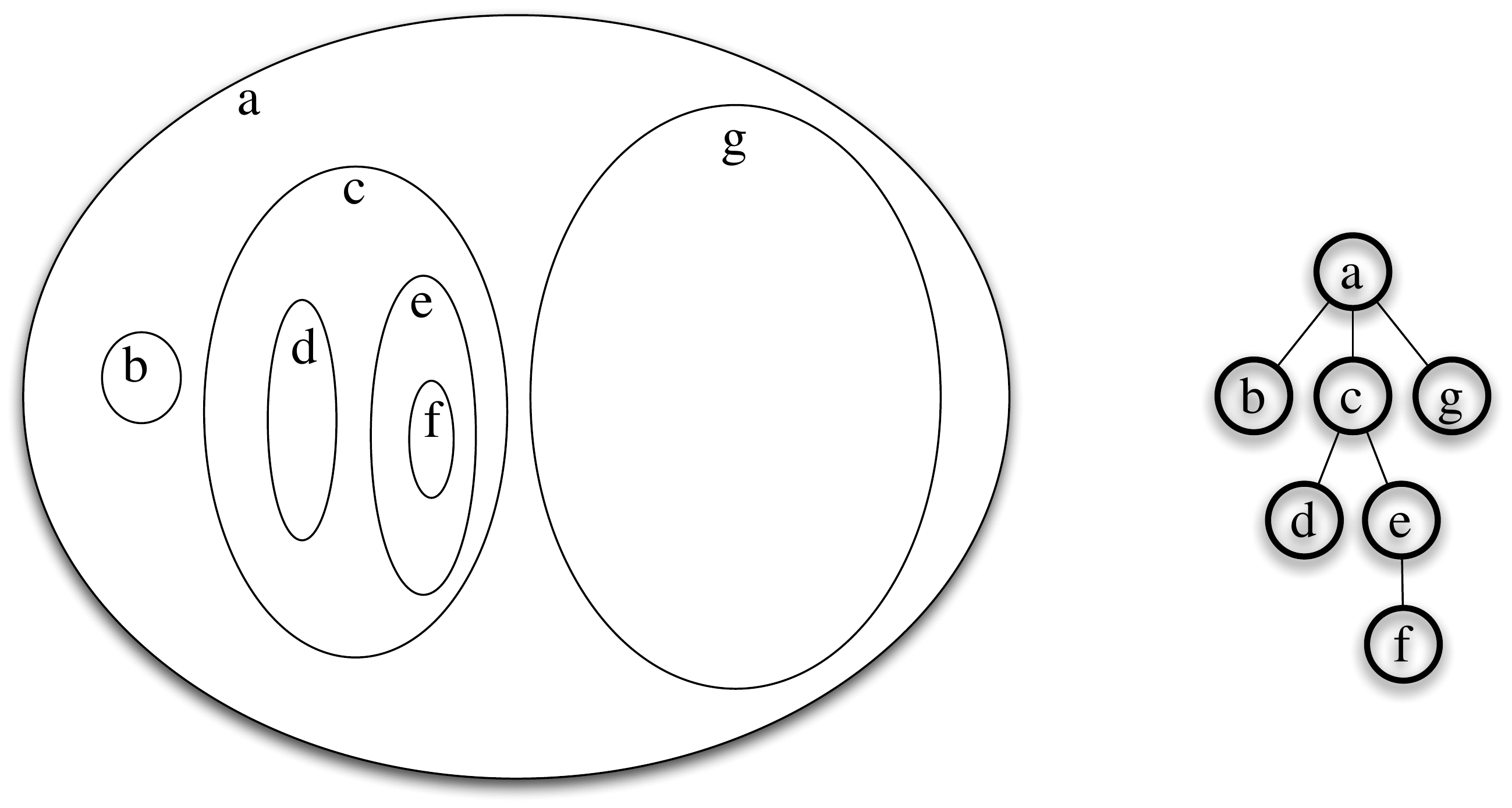}
\caption{On the left are non-crossing cycles.  On the right is the
  corresponding tree. Siblings are ordered left to right in increasing
  order of weight enclosed.  The cycles are labeled in preorder.}
\label{fig:preorder}
\end{figure}

The skeleton building algorithm is given as
Algorithm~\ref{alg:skeleton}. As the algorithm progresses, cycles and
regions are defined with respect to the {\em current} set of
cycles $\mathcal C$.

\newcommand{\ptr}{\text{\tt ptr}}
\begin{algorithm} \caption{Skeleton$(G,\lambda)$}
\label{alg:skeleton}
\begin{algorithmic}[1]
\STATE $\mathcal C \leftarrow \{\text{boundary of } f_\infty\}$
\STATE $\ptr \leftarrow \text{\it  null}$
\REPEAT

  \IF{some region $R$ with respect to $\mathcal C$ strictly contains a cycle $C$ s.t. $\ratio_R(C) \leq \epsilon^{-1} \lambda$} \label{line:if-blue}
    \STATE  let $C$ be a \textbf{maximally} face-enclosing such cycle
    \label{line:find-blue}
    \STATE set $F(C)$ to be the set of faces of $R$ enclosed by $C$
    \STATE add $C$ to $\mathcal C$ \label{line:add-blue}
  \ELSE
  \STATE let the cycles of $\mathcal{T}$ in preorder be $C_1 \dots C_k$
   \STATE \algorithmicif\ $\ptr=\text{\it null}$\ \algorithmicthen\
   $\ptr \leftarrow C_k$
   \STATE \algorithmicelse\ $\ptr \leftarrow$ cycle preceding $\ptr$
   in preorder of $\mathcal T$ 
   \STATE let $C_q$ be the value of $\ptr$\label{line:Cq}
    \STATE let $p$ be min such that $C_p$ is an ancestor of
    $C_q$ and $w_G(C_p)<2\,w_G(C_q)$
    \STATE remove cycles $C_{p+1}, \ldots, C_{q-1}$ from $\cal C$\label{line:splice}
  \ENDIF
\UNTIL $\ptr = \text{boundary of } f_\infty$
\STATE return $\mathcal{C}$
\end{algorithmic}
\end{algorithm}

Consider an execution of Line~\ref{line:splice} where a 
sequence of cycles $C_{p+1}, \ldots, C_{q-1}$ is removed from
$\mathcal C$. We say that $(C_p,C_q)$ is a \emph{splicing pair}.
Cycle~$C_p$ is the {\em rootward} member of the pair, and~$C_q$ is the
{\em leafward} member.

The following lemmas establish several useful properties of the skeleton, including the polynomial
running time of the skeleton-building algorithm.
\begin{lemma}\label{lem:fixed_suffix}
Let~$B$ be the value of $\ptr$ at any time $t$ in which Line~\ref{line:Cq} is executed. 
The suffix of the preorder starting at~$B$ does not change after time $t$.
In particular, cycle~$B$ will never be removed from the skeleton.
\end{lemma}
\begin{proof}

The proof is by induction on the number of times line~\ref{line:Cq} is
executed after time $t$.  If this number is zero, i.e. 
if time $t$ is the last execution of line~\ref{line:Cq}, then 
$\mathcal T$ can only change when a splice with splicing pair $(A, B)$
occurs, for some cycle $A$.
By definition of the algorithm, no new cycles can be added
to~$\mathcal T$ at time $t$.
Splicing with splicing pair $(A,B)$ does not change the region of any cycle 
to the left of $A$ (exclusive) or 
to the right of~$B$ (inclusive).
Therefore, no new cycles can be added to the left of $A$ or to the
right of $B$ after the splice.
In other words, any new cycle~$C$ added after the splice 
is added as a descendent of $A$, but not of $B$.
Since  $w_G(B) > \frac{1}{2}w_G(A)$, $B$ always lies on the rightmost path in the
subtree of $\mathcal T$ rooted at $A$. Hence any such new cycle $C$ is
added left of $B$ in the preorder.

For the inductive step, by the argument above, the suffix
starting at $B$ does not change until the next time
line~\ref{line:Cq} is executed. Since $\ptr$ moves strictly to the
left, the inductive hypothesis implies that the suffix remains
unchanged for the remainder of the execution.
\end{proof}

\begin{lemma}The skeleton algorithm runs in polynomial time.
\end{lemma}
\begin{proof}
By Lemma~\ref{lem:fixed_suffix}, each time line~\ref{line:Cq} is
executed, the length of the fixed suffix of the preorder of $\mathcal
T$ increases by one. The number of cycles in the skeleton is bounded
by the number of faces in $G$, so there are~$O(n)$ splicing steps.
Hence,
the total number of cycles added to $\mathcal
C$ throughout the execution of the skeleton construction algorithm is
$O(n^2)$. Since each cycle is found in polynomial time, the total
running time is polynomial.
\end{proof}

\begin{lemma}\label{lem:skeleton-cost}
The cost of the skeleton is $O(\eps^{-1} \opt \log W)$.
\end{lemma}
\begin{proof}
Let~$f$ be a face of~$G$, and let~$C(f) = \{C \in \mathcal{C} : f \in F(C)\}$.
By Lemma~\ref{lem:fixed_suffix}, every cycle~$A$ in the final skeleton is pointed to by
$\ptr$ at some time~$t$.
Cycle~$A$ then participates in a (possibly trivial) splice, removing all but one ancestor cycle
with weight less than twice~$\weight_G(A)$.
Afterward, for each cycle~$C$ added to~$\mathcal{C}$, we have~$F(C) \cap F(A) = \emptyset$.
Let~$C_1,\dots,C_k$ be all cycles of~$C(f)$ in rootward ordering.
The weight of the cycles in the sequence doubles at least once every other cycle.
We have~$|C(f)| = O(\log W)$.

The ratio of each cycle $C$ in the skeleton is bounded by $\eps^{-1} \lambda$, so
$\cost(C) \leq \eps^{-1}\lambda \weight_G(F(C))$.
Thus, the total cost of cycles in the skeleton is at most
\begin{align*}
\sum_C \eps^{-1}\lambda \weight_G(F(C))
&\leq \eps^{-1}\lambda \sum_{f \in G} |C(f)| \cdot \weight(f)\\
&\leq O(\eps^{-1} \lambda  W \log W)\\
&= O(\eps^{-1}  \opt \log W).
\end{align*}
\end{proof}

The heavy nesting of low-ratio cycles in the skeleton guarantees the following lemma, which will
be crucial in proving our spanner construction is correct.
\begin{lemma}\label{lem:no-low-ratio-remaining}
Suppose that, when the algorithm terminates, a region $R$ 
contains cycle $C\not\in \mathcal C$, and $\ratio_G(C) \leq \eps^{-1} \lambda$. 
Then $C$ encloses the heaviest hole of $R$, and
$\ratio_S(C)>\epsilon^{-1} \lambda$ where $S = $ \{faces enclosed
  by outer boundary of $R$ but not by heaviest hole\}
\end{lemma}
\begin{proof} We say a cycle $D$ {\em weakly crosses} $C$ if $D$
  crosses $C$ or $D=C$.  The following claim is immediate:

{\bf Claim 1:} At termination, $\mathcal C$ contains no cycle that weakly crosses $C$.

{\bf Claim 2:} Some splice removes a cycle weakly crossing $C$.

{\bf Proof of Claim 2:}
At some point the algorithm adds to $\mathcal C$ a cycle that weakly crosses $C$ or is strictly enclosed by
  $C$.  (If this never happens then $C$ remains available to choose in
  line~5, and the algorithm is not ready to terminate.)  Furthermore,
  by maximality in line~5, before the algorithm adds to $\mathcal C$
  any cycle strictly enclosed by $C$, it must add a cycle that
  weakly crosses $C$.  By Claim 1, such a cycle must have been removed
  by a splice.\qed

{\bf Claim 3:} For the last splice removing a cycle $D$ that weakly
crosses $C$, the splicing pair $(A,B)$ must satisfy the following
condition:\\
{\large \bf\raisebox{-.5ex}{*}}\ \ each
  edge of $C$ is enclosed by $A$ and not strictly enclosed by $B$.

{\bf Proof of Claim 3:} Since the splice removes
$D$, every edge of $D$ is enclosed by $A$ and not strictly enclosed
by $B$.  After
that splice, $A\in \mathcal C$.  If $C$ weakly crosses $A$ then by Claim~1
$A$ must be removed later, contradicting the choice of $(A,B)$.
Similarly, $C$ does not weakly cross $B$.  This implies every edge of
$C$ is enclosed by $A$ and is not strictly enclosed by $B$, proving
the claim. \qed

Claims~2 and~3 imply that there is some splice whose splicing pair
satisfies Condition~*.    Let $(A^0,B^0)$ be the splicing pair of the {\em
  last} such splice.  Let $t_0$ be the time of that splice.  

  Let $S_0 = $ \{faces enclosed by $A^0$ and not by $B^0$\}.  We claim
  $\ratio_{S_0}(C) > \epsilon^{-1} \lambda$.   If not then after time $t_0$
  the cycle $C$ would still
  be available to add to $\mathcal C$, so some cycle weakly crossing
  $C$ would have been added, contradicting either Claim 1 or the
  definition of $t_0$.  Because $\ratio_S(C)>\ratio_G(C)$ and $B^0$
  strictly contains no edge of $C$, it follows that $C$ encloses $B^0$.
By Lemma~\ref{lem:fixed_suffix}, cycle $B^0$ is never subsequently
removed from $\mathcal C$.  Also, by Claims~1 and~3 no cycle weakly
crossing $C$ is subsequently added to $\mathcal C$.  Therefore, for
every $t\geq t_0$, at time $t$ there is a unique pair $A^t,B^t$ of cycles in $\mathcal C$
such that $A^t$ encloses $C$, and $C$ encloses $B^t$, and $A^t$ is the
parent of $B^t$ in the tree $\mathcal T$ of cycles of $\mathcal C$.

Since $B^0$ remains in $\mathcal C$, we infer $B^t$ encloses
$B^0$ for all $t \geq t_0$.
 We show by induction that at each time $t \geq t_0$,
cycle $A^0$ encloses $A^t$.
There are two mutually exclusive cases.
First, suppose that just after time $t$ a cycle $A^{t+1}$ enclosing $C$ is added to
$\mathcal C$ and becomes the parent of $B^t$.   Then, since $A^t$ is
still in $\mathcal C$, $A^t$ encloses $A^{t+1}$ and $A^0$ encloses $A^{t+1}$.

Second, we show that no splice at time~$t+1$ can remove $A^t$.  Assume for
a contradiction that some splice did remove $A^t$, and let $(A,B)$ be the splicing pair.  
If $B$ were not a descendant of $A^t$ in $\mathcal T$ then the fact that the splice
removes $A^t$ would mean it would remove all descendants of $A^t$,
including $B^0$, a contradiction.  If $B$ crosses $C$ then by Claim~3 the splice or a later one
satisfies Property~*, contradicting the choice of $t_0$.
Therefore $B$ is a proper descendant of $A^t$ that does not cross
$C$.  $B$ cannot enclose $C$ else it would be a parent of $B^t$,
contradicting the fact that $A^t,B^t$ are a parent-child pair.
Therefore $B$ strictly encloses no edge of $C$.  But then $(A, B)$
satisfy Property~*, a contradiction.

Let $T$ be the time the algorithm terminates.
We have shown that $A^0$ encloses $A^T$ and $B^T$ encloses $B^0$.  
Let $S=$ \{faces enclosed by $A^T$ and not by $B^T$\}.  It follows
that $w_G(B^T)\geq w_G(B^0)$ and $w_G(A^T) \leq w_G(A^0)$ and $S
\subseteq S_0$.    Because
$w_G(A^0) < 2 w_G(B^0)$, we have $w_G(A^T) < 2 w_G(B^T)$, so $B^T$ is
the heaviest hole of the region whose outer boundary is $A^T$.
Because $S \subseteq S_0$ and $\ratio_{S_0}(C) > \epsilon^{-1}
\lambda$, it follows that $\ratio_S(C) > \epsilon^{-1} \lambda$.  This
completes the proof. \qed

\end{proof}

%% file: simple-shortcuts.tex
\section{Shortcuts}

Consider an optimal solution~$O$.
Let~$K$ be a cycle of~$O$ that crosses the skeleton.
A \emph{path decomposition} of~$K$ is a decomposition of~$K$ into paths $p_0, p_1, \dots$ such
that the number of paths is minimized and no path $p_i$ crosses the skeleton.
Because the number of paths is minimized, the endpoints of the paths occur only on skeleton cycles
that are crossed by~$K$.

Let $p$ be a member of a path decomposition as defined above.
Let $u$ and $v$ be the endpoints of~$p$ and let $R$ be the  region of the skeleton that
contains $p$.
In order to build the spanner, we wish to replace $p$ by a $(1+\eps)$-approximately shortest
$u,v$-path~$p'$ in~$R$.
For concision, we call the approximately shortest path $p'$ a \emph{shortcut}.
The spanner contains a set of edges that carry shortcuts for all paths of all path decompositions.
These edges are chosen so their total cost is sufficiently small and replacing each path by its
shortcut does not perturb the weight enclosed by the solution by more than a small amount.

If~$R$ contains no holes, then it is relatively easy to compute shortcuts within~$R$.
Let~$C$ be the outer boundary of~$R$.
The spanner algorithm computes a subgraph~$A$ of~$G$ within~$R$ using the boundary-to-boundary
spanner of Lemma~\ref{lem:boundary-spanner}.

As we show later in Lemma~\ref{lem:replace-paths-simple}, replacing paths through regions without
holes by their shortcuts does not change the weight of the solution by very much.
Essentially, the faces that change sides when a path~$p$ is replaced by its shortcut~$p'$ are
exactly those enclosed by~$p \concat \rev(p')$.
The cycle enclosing these same faces must have low weight or there would exist a low ratio cycle
within~$R$ violating Lemma~\ref{lem:no-low-ratio-remaining}.

\begin{figure*}
\label{fig:surgery}
\centering
\setlength{\tabcolsep}{0.75in}
\begin{tabular}{c c}
\includegraphics[height=1.5in]{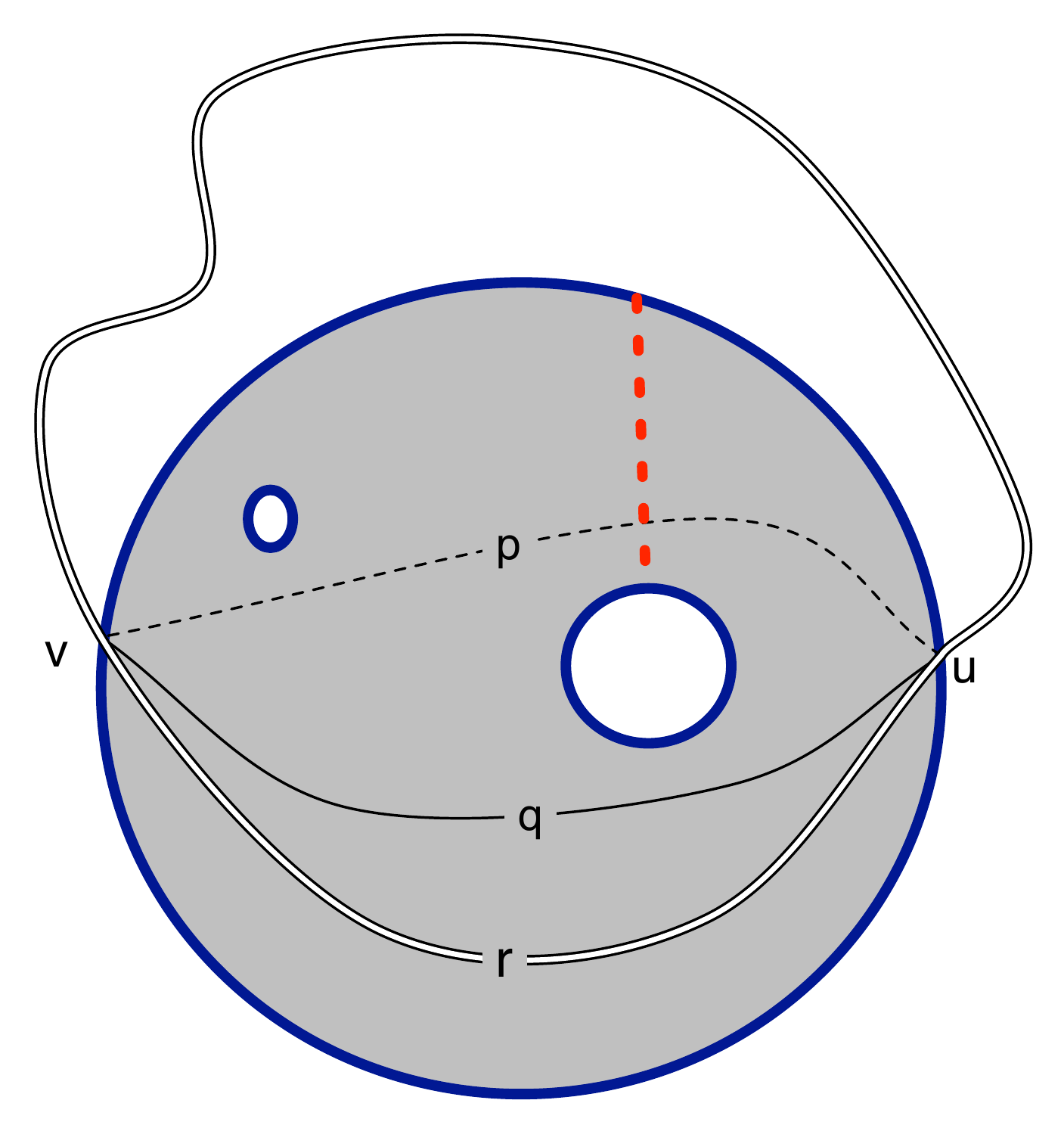} &
\includegraphics[height=1.5in]{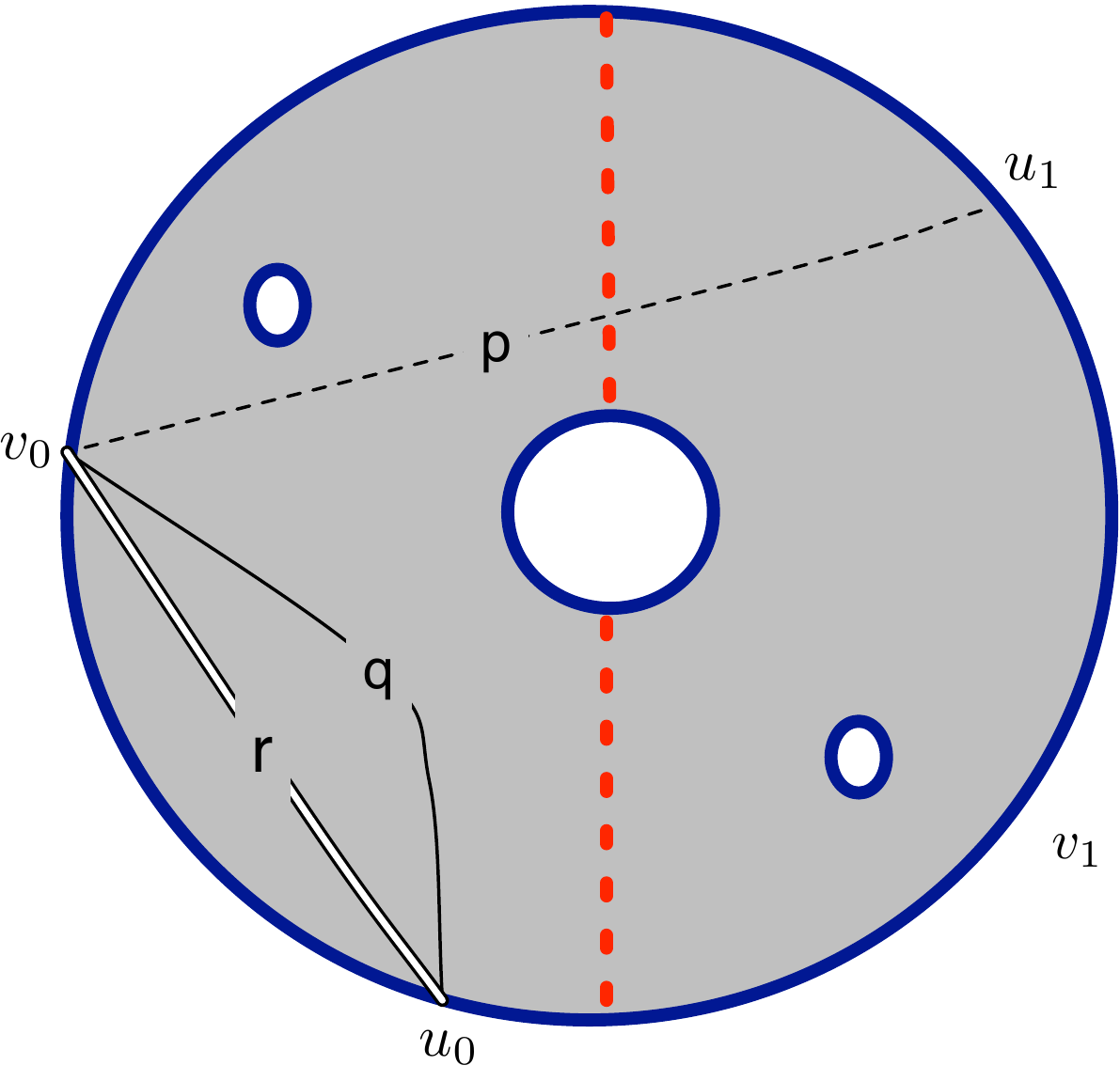}\\
(a) & (b)
\end{tabular}
\caption{(a) Path $r$ of a path decomposition enters and leaves a  region at vertices $u$
and $v$.
Replacing $r$ with shortcut $q$ will only affect the enclosure of faces outside the largest hole
of the region.
Replacing $r$ with shortcut $p$ will affect faces within the hole as well.
(b) The cyclic double cover of the region.
Lifts of $r$ and $q$ have the same endpoints in the cyclic double cover, but no lift of $p$ has
the same two endpoints.
Therefore, $p$ would not be a considered a shortcut for $r$ within the cyclic double cover.
}
\end{figure*}

Unfortunately, this argument does not hold if~$R$ contains holes.
First, there may be paths in a path decomposition that start and end on different boundary
components of~$R$.
A boundary-to-boundary spanner does not contain shortcuts for such paths.
And even if shortcuts are available,~$p \concat \rev(p')$ could enclose a high weight hole,
and all of the faces within that hole could change sides after replacing $p$ with $p'$.
See Figure~\ref{fig:surgery} (a).
Lemma~\ref{lem:no-low-ratio-remaining} does not imply anything about a cycle's weight if
it encloses the largest hole of a region, so there would be no limit to how much weight could
change sides.
In the next section, we describe how to address both of these issues for regions with holes.
We begin by describing a topological construct called the cyclic double cover used to address the
later issue.
A procedure called PC-clustering will be used to address the former.

%% file: double-cover.tex
\section{The cyclic double cover}

In this section, we describe a tool called the {\em cyclic double cover} that is used by our
spanner algorithm.
Thanks to this tool, the spanner will carry shortcuts that do not force the largest hole in
each region to change sides.
The cyclic double cover as described here was originally used by
Erickson~\cite{Erickson2011NontrivialCycles} and by Fox~\cite{Fox2013NontrivialCycles} to find short
topologically interesting cycles in surface embedded graphs.
Our presentation of the cyclic double cover is based closely on theirs.

Let $R$ be a region with outer boundary $C_0$ and largest hole $C_i$.
Let $L$ be an arbitrary path from $C_0$ to $C_i$.
Cut along $L$ to create a new region $R'$ where $C_0$, $C_i$, and two copies of $L$ ($L^+$
and $L^-$) form the outer boundary.
Let $(R', 0)$ and $(R', 1)$ be two distinct copies of $R'$.
For any vertex $v$ in $R$, let $(v,0)$ denote the copy of $v$ in $(R',0)$ and let $(v,1)$ denote
the copy in $(R',1)$.
Finally, let $(L^\pm,0)$ denote the copies of $L^\pm$ in $(R',0)$ and $(L^\pm, 1)$ denote the
copies in $(R',1)$.
The cyclic double cover $R^2$ is the planar graph resulting from identifying $(L^+,0)$ and
$(L^-,1)$ to a single path $(L,0)$ and identifying $(L^+,1)$ and $(L^-,0)$ to $(L,1)$.
Every hole of $R$ appears twice as a hole in $R^2$ except for $C_i$.
The edges of $C_0$ appear twice along the outer face of $R^2$ and the edges of $C_i$
appear twice along a single hole in $R^2$.
See Figure~\ref{fig:surgery} (b).

The cyclic double cover $R^2$ has an equivalent combinatorial definition.
Each vertex $v$ from $R$ has two copies $(v,0)$ and $(v,1)$ in $R^2$.
Edge $uv$ has two copies $(u,0) (v,0)$ and $(u,1) (v,1)$ if $uv$ does not enter~$L$ from the left.
Otherwise, the copies of $uv$ are $(u,0) (v,1)$ and $(u,1) (v,0)$.
Edge copies in~$R^2$ retain their costs from~$R$.
The \emph{projection} of any vertex, edge, or walk in~$R^2$ is the natural map to $R$ that
occurs by dropping the $0$ or $1$ from the vertex and edge tuples.
We say a vertex, edge, or walk $p$ in~$R$ \emph{lifts} to~$p'$ if~$p$ is the projection of~$p'$.
The outer boundary of~$R^2$ projects to two copies of~$C_0$ and one boundary of~$R^2$ projects
to two copies of~$C_i$.
Otherwise, every face or boundary of~$R^2$ projects to a face or hole in~$R$.
Call a boundary in~$R^2$ a hole if it projects to one or two copies of a hole in~$R$.
For a walk~$p$ in $R$, let~$x^2(p) = 0$ if $p$ crosses $L$ an even number of times.  Otherwise,
let $x^2(p) = 1$.
We immediately get the following lemmas.
\begin{lemma}\label{lem:walk-lift}
Let $p$ be a walk in $R$ from vertex $u$ to vertex $v$.
Walk $p$ is the projection of a unique walk $p'$ in $R^2$ from vertex $(u,0)$ to vertex
$(v,x^2(p))$.
\end{lemma}
\begin{lemma}\label{lem:closed-walk-projection}
Let~$C'$ be a closed walk in~$R^2$.
$C'$ projects to a unique closed walk~$C$ in~$R$ that does not enclose $C_i$.
\end{lemma}

Our algorithm builds a spanner by computing the cyclic double cover for each region~$R$
of the skeleton that contains a hole.
The construction of the double covers can be done in linear time.
It then computes a set of edges within each double cover that carry shortcuts for the lifts of
paths that may appear in path decomposition.
See Figure~\ref{fig:surgery} (b).
The projections of these edges back into the original regions will be added to the spanner.
In Lemma~\ref{lem:replace-paths}, we will show that we can replace the paths of path
decompositions by the projections of these shortcuts without changing the weight of the solution
by more than a small amount.

%% file: clustering.tex
\section{PC-clustering}

Our spanner algorithm needs to find edges that carry shortcuts within each cyclic double cover.
However, the\newline
boundary-to-boundary spanner of Lemma~\ref{lem:boundary-spanner} will not find edges
to carry shortcuts between distinct boundaries of a region.
In order to use the algorithm, we augment the skeleton's edges within each double cover using the
PC-clustering algorithm of Bateni, Hajiaghayi, and Marx~\cite{BateniHM11}.

Let~$K$ be a cycle in the optimal solution that crosses the skeleton, and let $p$ be a path of
$K$'s path decomposition where $p$ lies in a region~$R$ with a hole.
The PC-clustering algorithm adds a relatively cheap set of edges to the boundary of the cyclic
double cover~$R^2$.
These edges are chosen so that, \emph{in general}, both of~$p$'s endpoints lie on the same
boundary component after running PC-clustering.
If $p$'s endpoints are still on different components, then at least one of the components must be
very cheap.
The edges of that component can be added to the optimal solution without substantially increasing
its cost, and $K$ can be modified to avoid crossing the skeleton cycles in that component.
Path~$p$ is no longer in a path decomposition, and it is no longer necessary for the spanner to
hold a shortcut between $p$'s endpoints.
Formally, the PC-clustering algorithm can be described as follows.

\begin{longversion}
\begin{lemma}[Bateni \etal~\cite{BateniHM11}]\label{lem:PC-clustering}
Let $G(V,E)$ be a graph with non-negative edge costs $\cost(e)$ 
and face potentials $\phi(v)$.
There exists a polynomial time algorithm to find a subgraph $Z$ such that
\begin{enumerate}
\item
the total cost of $Z$ is at most $2 \sum_{v \in V} \phi(v)$ and
\item
for any subgraph $H$ of $G$, there is a set $U$ of vertices such that
  \begin{enumerate}
  \item
  $\sum_{v \in U} \phi(v)$ is at most the cost of $H$ and
  \item
  if two vertices $v_1,v_2 \notin U$ are connected by $H$, then they are in the same component of
  $Z$.
  \end{enumerate}
\end{enumerate}
\end{lemma}
\end{longversion}

\begin{shortversion}
\begin{lemma}[Bateni \etal~\cite{BateniHM11}]\label{lem:PC-clustering}
Let $G(V,E)$ be a graph with non-negative edge costs $\cost(e)$ 
and face potentials $\phi(v)$.
There exists a polynomial time algorithm to find a subgraph $Z$ such that
(1) the total cost of $Z$ is at most $2 \sum_{v \in V} \phi(v)$ and
(2)
for any subgraph $H$ of $G$, there is a set $U$ of vertices such that
  (2a)
  $\sum_{v \in U} \phi(v)$ is at most the cost of $H$ and
(2b)
  if two vertices $v_1,v_2 \notin U$ are connected by $H$, then they are in the same component of
  $Z$.
\end{lemma}
\end{shortversion}

For each region~$R$ containing a hole, our algorithm does the following:
It contracts the lifts of every skeleton edge in $R^2$ to get the
graph~$\hat{R}^2$.
For any vertex $v$ in $\hat{R}^2$, let $\cost(v)$ be the total cost of edges contracted to
create $v$ (implicitly, if $v$ appears in $R^2$ as well, then $\cost(v) = 0$).
For each $v$ in $\hat{R}^2$, the algorithm assigns a potential $\phi(v) = \eps^{-1} \cost(v)$.
The algorithm them applies the PC-clustering procedure of Lemma~\ref{lem:PC-clustering} to
$\hat{R}^2$ to get the set of edges~$Z$.

Let the \emph{well-connected cover graph} be the set of boundary in $R^2$ unioned with $Z$.
The edges in the well-connected cover graph are the \emph{well-connected cover edges}.
The projections of the well-connected cover edges will be used in the spanner.
\begin{lemma}\label{lem:clustering-edges-cheap}
The total cost of all well-connected cover edges is $O(\eps^{-2} \opt \log W)$.
\end{lemma}
\begin{proof}
By Lemma~\ref{lem:skeleton-cost}, the total cost of all cycles in the skeleton is~$O(\eps^{-1}
\opt \log W)$.
Every cycle in the skeleton appears on the boundary of at most two regions.
For each boundary of a region, each edge of the boundary appears at twice in that
region's cyclic double cover.
Therefore, the sum of vertex potentials used for PC-clustering across all cyclic double covers is
at most~$O(\eps^{-2} \opt \log W)$.
The lemma follows from the first property of PC-clustering's output as defined in
Lemma~\ref{lem:PC-clustering}.
\end{proof}

We argue there exists a near-optimal solution such that shortcuts do not start and end on
different components of the well-connected cover graphs.
\begin{shortversion}
\begin{lemma}\label{lem:common-boundary-endpoints}
There exists a solution with cost at most $(1+2\eps)\opt$ enclosing exactly $bW$ weight such that
for each cycle $K$ in the solution,
either $K$ does not cross the skeleton
or $K$ has a path decomposition $p_0,p_1,\dots$ such that a lift of each path $p_i$ in a 
region with a hole has both endpoints on the same component of that region's well-connected
cover graph.
\end{lemma}
\end{shortversion}
\begin{longversion}
\begin{lemma}\label{lem:common-boundary-endpoints}
There exists a solution with cost at most $(1+2\eps)\opt$ enclosing exactly $bW$ weight such that
for each cycle $K$ in the solution,

\begin{enumerate}
\item
either $K$ does not cross the skeleton
\item
or $K$ has a path decomposition $p_0,p_1,\dots$ such that a lift of each path $p_i$ in a region with a hole has both endpoints on the same component of that region's well-connected
cover graph.
\end{enumerate}
\end{lemma}
\end{longversion}
\begin{proof}
Consider the optimal solution $O$.
For each region $R$ with a hole, let $O_R$ be the subset of edges from $O$ that lie
strictly within the region $R$.
Consider a lift of $O_R$ to $R^2$, and let $\hat{O}_R$ be the edges that remain after performing
contractions to get $\hat{R}^2$.
Let $\hat{U}_R$ be the set of vertices in $\hat{R}^2$ that are guaranteed to exist for $\hat{O}_R$ by the
second property in Lemma~\ref{lem:PC-clustering}.


Each vertex of $\hat{U}_R$ is the result of contracting zero or more lifted skeleton cycles in
$R^2$.
Let $\mathcal{U}$ be the set of all skeleton cycles where for each cycle $C \in \mathcal{U}$, the
contraction of a lift of $C$ lies in some $\hat{U}_R$.
The cycles in $\mathcal{U}$ are mutually non-crossing, so they partition the faces of~$G$ into
\emph{$\mathcal{U}$-regions}.
For each $\mathcal{U}$-region $R'$, let~$O_{R'}$ be the boundary of faces in~$R'$ enclosed by~$O$
(therefore, the holes of~$R'$ are not enclosed by~$O_{R'}$).
Let~$O'$ be the union of cycles over all~$O_{R'}$.
Solution~$O'$ encloses the same set of faces of~$G$ as~$O$.
No cycle of~$O'$ crosses a member of~$\mathcal{U}$.
Edges strictly internal to some $\mathcal{U}$-region are used exactly once and only if they are
used in~$O$.
Finally, each cycle~$C \in \mathcal{U}$ may contribute up to two copies of some of its edges
to~$O'$, because $C$ lies on the boundary of two $\mathcal{U}$-regions.

The cost difference between $O'$ and $O$ is at most twice the cost of cycles in $\mathcal{U}$.
For a single region $R$, the total potential of vertices in $\hat{U}_R$ is at most the cost
of $\hat{O}_R$.
By definition of vertex potentials, the cycles of $\mathcal{U}$ that contribute to set $\hat{U}_R$
have total cost at most $\eps \cost(\hat{O}_R)$.
All sets $\hat{O}_R$ are edge-disjoint, so twice the total cost of all cycles in $\mathcal{U}$ is at
most $2 \eps  \opt$.

Now, consider any cycle $K$ in $O'$ that crosses the skeleton, and let $p_0,p_1,\dots$ be a path
decomposition for $K$.
The endpoints of each path $p_i$ lie on skeleton cycles crossed by $K$.
For any $p_i$ lying in a region~$R$ with a hole, let $p'_i$ be a lift of path $p_i$
to $R^2$.
Let $p^/_i$ be the path that results from $p'_i$ after contracting edges to make
$\hat R^2$.
Finally, let $v_1$ and $v_2$ be the endpoints of $p^/_i$.
Neither $v_1$ nor $v_2$ lie on a member of $\hat{U}_R$ as defined above, because the cycles of
$O'$ do not cross any members of $\mathcal{U}$.
Further, every edge of $p'_i$ that is not contracted is a member of $\hat{O}_R$.
Therefore, Lemma~\ref{lem:PC-clustering} guarantees $v_1$ and $v_2$ lie in the same component of
PC-clustering's output.
Each vertex of $\hat R^2$ is a connected component of the boundary of $R^2$ so the endpoints of
$p'_i$ lie on the same component of the well-connected cover graph as well.
\end{proof}

\begin{longversion}
\subsection{Finding shortcuts}
\end{longversion}

For each cyclic double cover $R^2$, for each well-connected cover component, our spanner
algorithm computes edges carrying shortcuts between every pair of vertices on the component.
It does so using the following extension of the boundary-to-boundary spanner of
Lemma~\ref{lem:boundary-spanner}.
The projection of these edges is added to our spanner.
\begin{shortversion}
The proof of the following Lemma is omitted.
\end{shortversion}
\begin{lemma}\label{lem:boundary-to-boundary}
Let $G$ be a planar graph with non-negative edge costs $\cost(\cdot)$.
Let $A$ be a component of $G$.
Then for any $\eps > 0$, there is an $O(n \log n)$ time algorithm to compute a
subgraph $H$ of $G$ where $\cost(H) = O(\eps^{-4} \cost(A))$ and for any pair of vertices $u$ and $v$ on
$A$, $\dist_H(u,v) \leq (1+\eps)\dist_G(u,v)$.
\end{lemma}
\begin{longversion}
\begin{proof}
The edges and vertices of~$A$ partition the plane into one or more components.
The algorithm cuts the planar graph~$G$ along $A$, separating these components so that the edges
of~$A$ lie on their boundary.
For each boundary cycle $C$ of the cut open graph, the algorithm runs the boundary-to-boundary
spanner procedure of Lemma~\ref{lem:boundary-spanner} to create a subgraph $H_C$ of $G$ of cost at
most $O(\eps^{-4} \cost(C))$ in $O(n_C \log n_C)$ time where $n_C$ is the number of vertices in
$C$'s component of the cut open planar graph.
Subgraph $H$ is the union of all such subgraphs $H_C$.
Each edge of $A$ appears on boundary twice, so the total cost of all subgraphs $H_C$ is at most
$O(\eps^{-4} \cost(A))$.
The cut open surface has $O(n)$ vertices total, so the running time for the procedure is $O(n
\log n)$ total.
Finally, the boundary-to-boundary shortcut algorithm guarantees that for any pair of vertices $u$
and $v$ on $A$ where the shortest path does not cross any cycle $C$ of the cut open surface, we
have $\dist_H(u,v) \leq (1+\eps)\dist_G(u,v)$.
This proves the lemma since a shortest path that does cross a  cycle $C$ of the cut open
surface is the concatenation of shortest paths that do not cross any
such cycle.
\end{proof}
\end{longversion}

\begin{lemma}\label{lem:shortcuts-cheap}
The total cost of all edges carrying shortcuts is $O(\eps^{-6} \opt \log W)$.
\end{lemma}
\begin{proof}
By Lemma~\ref{lem:skeleton-cost}, the total cost of boundaries for regions without a 
hole is $O(\eps^{-1} \opt \log W)$.
By Lemma~\ref{lem:clustering-edges-cheap}, the total cost of all subgraphs $A$ used in
Lemma~\ref{lem:boundary-to-boundary} for regions with holes is $O(\eps^{-2} \opt
\log W)$.
The current lemma follows from Lemmas~\ref{lem:boundary-spanner}
and~\ref{lem:boundary-to-boundary}.
\end{proof}

%% file: spanner.tex
\section{The spanner}

In this section, we describe the final spanner construction for our algorithm and prove the
construction follows the spanner properties.
The spanner construction is summarized as Algorithm~\ref{alg:spanner}.
\begin{algorithm} \caption{Spanner$(G)$}
\label{alg:spanner}
\begin{algorithmic}[1]
\STATE $\mathcal{S} \leftarrow \emptyset$ ; 
$\mathcal{C} \leftarrow \text{Skeleton}(G,\lambda)$
\STATE add to $\mathcal{S}$ the edges of~$\mathcal{C}$ \label{line:include-skeleton}
\STATE for each region~$R$ without a hole, add to $\mathcal{S}$ the edges carrying shortcuts
in~$R$ \label{line:shortcuts-simple}
\STATE for each region~$R$ with a hole, add to $\mathcal{S}$ the projection of the well-connected
cover edges from~$R^2$ \label{line:well-connected}
\STATE for each region~$R$ with a hole, add to $\mathcal{S}$ the projection of the edges carrying
shortcuts in~$R^2$ \label{line:shortcuts-cover}
\STATE for each region~$R$ with a hole, add to $\mathcal{S}$ the edges on the cheapest cycle
in~$R$ that encloses the heaviest hole of~$R$ other than the outer boundary of~$R$
\label{line:shortest}
\end{algorithmic}
\end{algorithm}
The cheapest cycle other than the outer boundary enclosing a particular hole of a region~$R$ can
be computed in polynomial time using several instantiations of any polynomial time minimum
$s,t$-cut algorithm~\cite{ItalianoNSW11}.

To prove that our algorithm computes a spanner, we will iteratively replace cycles and their paths
in a near-optimal solution with ones that lie in the spanner.
The following lemmas will help us bound the total change in weight and cost from performing these
operations.

\begin{lemma}\label{lem:replace-empty-cycle}
Let $O$ be a set of cycles.
Let $K \in O$ be contained by region~$R$, and let $\ratio_G(K) > \eps^{-1}\lambda$.
Removing $K$ from $O$ changes the enclosed weight of $O$ by at most $\eps \cost(K)/\lambda$ and
does not increase the cost of $O$. 
\end{lemma}
\begin{lemma}\label{lem:replace-enclosing-cycle}
Let $O$ be a set of cycles.
Let $K \in O$ be strictly contained by region $R$, and let $\ratio_G(K) \leq \eps^{-1}\lambda$.
Cycle $K$ encloses the largest hole of $R$.
Further, replacing $K$ with the shortest cycle other than the outer boundary of~$R$ that encloses
the largest hole of~$R$ changes the weight enclosed by $O$ by at most $\eps \cost(K)/\lambda$ and
does not increase the cost of $O$.
\end{lemma}
\begin{proof}
By Lemma~\ref{lem:no-low-ratio-remaining}, we know $K$ must enclose the largest hole of $R$.
Let $A$ be the outer boundary of $R$ and $B$ be the largest hole of $R$.
Let~$K'$ be the shortest cycle other than~$A$ enclosing~$B$ in~$R$.
Replacing $K$ by a possibly cheaper cycle $K'$ cannot increase the cost of $O$.
Let $\Delta \weight(K)$ be the total weight of faces that move from being enclosed to not
enclosed and vice versa after the replacement.
Let $\tilde\Delta \weight(K)$ (respectively $\tilde\Delta \weight(K')$) be the total weight of
faces that are enclosed by $K$ ($K'$) but not enclosed by~$B$.
Suppose~$\tilde\Delta \weight(K) \geq \tilde\Delta \weight(K')$.
By Lemma~\ref{lem:no-low-ratio-remaining}, we have $\cost(K) / \tilde\Delta \weight(K) > \eps^{-1}
\lambda$.
In particular, $\Delta \weight(K) \leq \tilde\Delta \weight(K) < \eps \cost(K) / \lambda$.
If $\tilde\Delta \weight(K') > \tilde\Delta \weight(K)$, then either~$\tilde\Delta \weight(K') =
0$ or~$K'$ is not in $\mathcal C$.
Either way, we have $\cost(K') / \tilde\Delta \weight(K') > \eps^{-1} \lambda$, and $\Delta
\weight(K) \leq \tilde\Delta \weight(K') < \eps \cost(K') / \lambda \leq \eps \cost(K) / \lambda$.
\end{proof}
\begin{lemma}\label{lem:replace-paths-simple}
Let $O$ be a set of cycles.
Let $K \in O$ cross one or more cycles of the skeleton.
Let $p$ be a path of $K$'s path decomposition lying in region $R$ such that $R$ has no holes and
$p$ has at least one edge disjoint from the spanner.
Let $u$ and $v$ be the endpoints of~$p$ on the boundary of $R$.
Finally, let $p'$ be a shortcut in the spanner between $u$ and $v$ in $R$.
Replacing $p$ by $p'$ changes the weight enclosed by $O$ by at most
$3\eps \cost(p)/\lambda$ and increases the cost of $O$ by at most $\eps \cost(p)$.
\end{lemma}
\begin{proof}
The bound on the cost change follows from the fact that shortcut $p'$ is a $(1+\eps)$-approximate
shortest path.
Let $e$ be the edge of $p$ strictly internal to~$R$.
Let~$C = p \concat \rev(p')$.
Let~$C'$ be the cycle using darts of~$C$ that encloses and connects the boundary of faces of~$G$
enclosed by~$C$.
Cycle~$C'$ is strictly enclosed in~$R$, because it contains~$e$.
Also, $\cost(C') \leq \cost(C)$.
The faces of $G$ enclosed by $C'$ are exactly the faces that switch sides when replacing $p$ by
$p'$.
Therefore, the replacement will change the weight enclosed by $O$ by exactly $\weight_G(C')$.
By Lemma~\ref{lem:no-low-ratio-remaining}, we have $\ratio_G(C') > \eps^{-1} \lambda$.
Recall, a shortcut is a $(1+\eps)$-approximate shortest path.
Therefore,
$$
  \eps^{-1} \lambda < \frac{\cost(C')}{\weight_G(C')}
                    = \frac{\cost(p)+\cost(p')}{\weight_G(C')}
                    \leq \frac{(2+\eps)\cost(p)}{\weight_G(C')}.
$$
In particular, $\weight_G(C') < \eps(2+\eps)\cost(p)/ \lambda \leq 3\eps \cost(p) /
\lambda$.
\end{proof}
\begin{lemma}\label{lem:replace-paths}
Let $O$ be a set of cycles.
Let $K \in O$ cross one or more cycles of the skeleton.
Let $p$ be a path of $K$'s path decomposition lying in region $R$ such that $R$ has at least one
hole and $p$ has at least one edge disjoint from the spanner.
Let $p'$ be a lift of $p$ from $R$ to $R^2$ with endpoints $u$ and $v$.
Finally, let $p''$ be a shortcut in the spanner between $u$ and $v$ in $R^2$.
Replacing $p$ by the projection of $p''$ changes the weight enclosed by $O$ by at most $3\eps
\cost(p)/\lambda$ and increases the cost of $O$ by at most $\eps \cost(p)$.
\end{lemma}
\begin{proof}
Again, the bound on cost change is immediate.
Let~$C$ be the projection of~$p' \concat \rev(p'')$.
By Lemma~\ref{lem:closed-walk-projection}, $C$ does not enclose the largest hole of~$R$.
The rest of the proof is identical to that of Lemma~\ref{lem:replace-paths-simple}.
\end{proof}

We finally prove that our algorithm constructs a spanner.
\begin{lemma}\label{lem:spanner-containment}
The output of Algorithm~\ref{alg:spanner} contains a solution with cost at most $(1+4\eps) \opt$
enclosing $b'W$ weight with~$b' \in [b-6\eps,b+6\eps]$.
\end{lemma}
\begin{proof}
Let $O$ be the near-optimal solution given by\newline
Lemma~\ref{lem:common-boundary-endpoints}.
We will modify $O$ to create a solution $O'$ that lies in the spanner.

Consider any cycle $K$ in $O$.
If it matches the hypothesis of Lemma~\ref{lem:replace-empty-cycle}, then it is simply removed
from $O$.
If it matches the hypothesis of Lemma~\ref{lem:replace-enclosing-cycle}, then it is replaced by
the shortest cycle other than the outer boundary of its region enclosing the largest hole of its
region.
Line~\ref{line:shortest} of Algorithm~\ref{alg:spanner} guarantees this cycle exists in the
spanner.

In the final case, we know $K$ has a path decomposition $p_0,p_1,\dots$.
We replace each~$p_i$ in the decomposition as follows.
If~$p_i$ has no edges disjoint from the spanner, then it remains as is.
If~$p_i$ has an edge strictly internal to a region without holes, then we replace $p_i$ with a
$(1+\eps)$-approximate shortest path between its endpoints.
Line~\ref{line:shortcuts-simple} guarantees this path exists in the spanner.
Finally, if $p_i$ has an edge strictly internal to a region~$R$ with a
hole, then $p'_i$, $p_i$'s lift
to~$R^2$, has both endpoints on the same component of~$R$'s well-connected cover graph.
The output of the algorithm in Lemma~\ref{lem:boundary-to-boundary} contains an
$(1+\eps)$-approximate shortest path $p''_i$ in the cyclic double cover between the endpoints of
$p'_i$.
Lines \ref{line:well-connected} and \ref{line:shortcuts-cover} of Algorithm~\ref{alg:spanner}
guarantee the projection of $p''_i$ exists in the spanner.
Path~$p_i$ is replaced by the projection of $p''_i$.
Any remaining cycles or decomposition paths of~$O$ exist in the skeleton itself, which is also in
the spanner by Line~\ref{line:include-skeleton} of Algorithm~\ref{alg:spanner}.

Solution $O'$ is formed by performing the replacements described above.
Let~$p$ be any cycle or path replaced above. 
By
Lemmas~\ref{lem:replace-enclosing-cycle},~\ref{lem:replace-paths-simple},
and~\ref{lem:replace-paths}, 
after replacing~$p$, the cost of $O$ increases by at most $\eps \cost(p) / \lambda$, and the
weight of faces enclosed by~$O$ changes by at most $3\eps \cost(p)/\lambda$.
By summing over all cycles $K$, we see the total cost increases by at most $\eps(1+2\eps)\opt \leq
2\eps \opt$ (for $\eps \leq 1/2$) and the weight changes by at most $3\eps (1+2\eps) \opt /
\lambda \leq 6 \eps W$.
\end{proof}
\begin{lemma}\label{lem:spanner-cost}
The output of Algorithm~\ref{alg:spanner} has cost at most $O(\eps^{-6} \opt \log W)$.
\end{lemma}
\begin{proof}
The shortest cycles enclosing largest holes within their region have total cost at most that of
the skeleton itself.
The rest of the spanner's components have their costs bounded in Lemmas~\ref{lem:skeleton-cost},
\ref{lem:clustering-edges-cheap}, and \ref{lem:shortcuts-cheap}.
\end{proof}

%% file: ack.tex
\vspace{-2.5ex}
\paragraph{Final remarks}
 We have given a bicriteria approximation scheme, but there is much room for improvement.  Is there
an {\em efficient} PTAS, one whose running time is a polynomial with
degree independent of $\epsilon$?  Is there a PTAS that does not
approximate the balance $b$?  Is there perhaps even a polynomial-time
algorithm for bisection in planar graphs when the weights are small
integers?

The research described in this paper began at a meeting at Brown University's Institute
for Computational and Experimental Research in Mathematics.
The authors thank D\'aniel Marx who described to the second author some results on fixed-parameter
tractability of graph separation, leading to some discussion of the
complexity in planar graphs and  MohammadTaghi Hajiaghayi who asked whether there is an
approximation scheme for bisection in planar graphs.
Hajiaghayi had some initial thoughts due to its
similarity to multiway cut for which there is a PTAS on planar graphs.